\documentclass{ectj}

\usepackage{amsfonts,amssymb,graphics,epsfig,verbatim,bm,latexsym,amsmath,url,amsbsy}
\usepackage{booktabs}
\usepackage{enumitem}

\newtheorem{theorem}{Theorem}
\newtheorem{assumption}{Assumption}

\newtheorem{corollary}{Corollary}
\newtheorem{lemma}{Lemma}

\newtheorem{remark}{Remark}

\makeatletter
\@addtoreset{theorem}{section}
\@addtoreset{assumption}{section}
\@addtoreset{proposition}{section}
\@addtoreset{corollary}{section}
\@addtoreset{lemma}{section}
\@addtoreset{algorithm}{section}
\@addtoreset{example}{section}
\@addtoreset{remark}{section}
\makeatother

\DeclareMathOperator{\E}{\mathbb E}
\DeclareMathOperator{\Var}{Var}

\DeclareMathOperator{\tr}{tr}

\DeclareMathOperator{\diag}{diag}
\newcommand{\R}{\mathbb R}
\newcommand{\Pp}{\mathbb P}
\newcommand{\G}{\mathcal G}
\newcommand{\F}{\mathcal F}
\newcommand{\tX}{\widetilde X}
\newcommand{\hbeta}{\widehat\beta}
\newcommand{\hsigma}{\widehat\sigma}
\newcommand{\hOmega}{\widehat\Omega}
\newcommand{\hQ}{\widehat Q}

\newcommand{\htau}{\widehat\tau}

\title[Saturated FE Variance Estimation]{Variance Estimation for Saturated Fixed-Effect Specifications}

\author[Halkiewicz]{Stanis\l aw~M.~S.~Halkiewicz$^{\dagger}$}

\address{$^{\dagger}$Faculty of Applied Mathematics, AGH University of Krak\'ow, Poland.}
\email{smsh@student.agh.edu.pl}

\begin{document}

\begin{abstract}
We characterize the asymptotic behavior of conventional variance estimators in linear regression with high-dimensional fixed effects under a drift in which both the proportional fixed-effect dimension $\rho_n = d_{K_n}/n \to \rho \in [0,1)$ and the residual treatment variance $\tau_n^2 = nQ_{K_n} \to \tau^2 \in (0, \infty]$ are non-degenerate. Three findings emerge. First, under strict exogeneity and conditional homoskedasticity, the Cattaneo--Jansson--Newey-corrected $t$-statistic is asymptotically exact for any $\tau^2 > 0$: there is no Stock--Yogo-style threshold in $\tau^2$. Second, the Eicker--White HC0 estimator is biased downward by a fixed factor $(1-\rho)$, producing over-rejection that grows with saturation. Third, HC3 over-corrects in the opposite direction by a factor $1/(1-\rho)$. The leave-one-out estimator (HC2) removes the first-order leverage distortion and is asymptotically exact under homoskedasticity or design-balanced heteroskedasticity; under general heteroskedasticity with non-uniform leverage, HC2 retains an additional bias of order $\rho|\mu - \omega^2|$ that we characterize. An empirical application to Piotroski F-Score returns in CEE markets illustrates the predicted variance hierarchy in real data.

\end{abstract}

\section{Introduction}\label{sec-intro}

Applied econometric work routinely reports linear regressions with high-dimensional fixed effects: worker--firm models \citep{AKM, CHK2013, BLM2019}, three-way trade gravity equations, triple-differences with cell structure, event studies with unit-and-time-interacted controls. As fixed effects (FE) saturate, two things happen at once: the proportional dimension of the FE space, $\rho_n = d_{K_n}/n$, grows large, and the residual treatment variance, $Q_K = \E[\Var(X \mid \G_K)]$, shrinks. A natural conjecture, building on the weak-instrument literature \citep{SS, SY, Moreira2003, OleaPflueger2013}, is that small $Q_K$ produces a weak-identification problem requiring its own critical values: identification analogous to a small concentration parameter.

The conjecture turns out to be false in the baseline model. Under strict exogeneity and conditional homoskedasticity, FE-residualized OLS is unbiased regardless of $Q_K$, because the within transformation is an exact orthogonal projection rather than an estimated first stage. The Cattaneo--Jansson--Newey-corrected $t$-statistic, defined with the residual-based variance estimator $\widehat\sigma^2 = \widehat u' \widehat u/(n - d_{K_n} - 1)$, has an exact $t_{n - d_{K_n} - 1}$ distribution; its size distortion relative to $N(0,1)$ is of order $1/n$ and does not depend on the residual treatment variance $\tau_n^2 = nQ_{K_n}$. There is no Stock--Yogo-style threshold in $\tau^2$ for this setup, in contrast to the instrumental-variables case. This negative result is a useful piece of econometric clarification, and it redirects attention to where the real distortion in saturated FE specifications does come from: variance estimation.

We characterize the asymptotic behavior of conventional variance estimators under the joint drift $(\rho_n, \tau_n^2) \to (\rho, \tau^2)$ with $\rho \in [0,1)$ and $\tau^2 \in (0, \infty]$. Three findings emerge.

First, the standard (no degrees-of-freedom correction) homoskedastic variance estimator $\hsigma^2 = \widehat u' \widehat u / n$ produces $t$-tests with asymptotic size $2(1 - \Phi(z_{1-\alpha/2}\sqrt{1-\rho}))$, growing without bound in $\rho$. At $\rho = 0.5$ this is 17\% for a nominal-5\% test. The CJN correction $\hsigma^2 = \widehat u' \widehat u/(n - d_{K_n} - 1)$ restores asymptotic exactness.

Second, the Eicker--White (HC0) heteroskedasticity-robust estimator has an asymptotic bias toward zero of order $\rho$: $\hOmega_{\mathrm{HC0}}/(nQ_K) \to \omega^2(1-\rho) + \rho\,\E[\sigma^2(X, \G_\infty)]$, where $\omega^2$ is the relevant limiting score variance. Under homoskedasticity this collapses to $\sigma^2$ (so HC0 is consistent in the absence of heteroskedasticity), but under genuine heteroskedasticity it produces $t$-tests with growing over-rejection.

Third, and less appreciated, the conventional HC3 estimator over-corrects in the opposite direction. Under homoskedasticity and balanced FE structure, $\hOmega_{\mathrm{HC3}}/(nQ_K) \to \sigma^2 / (1-\rho)$, inflating the variance by factor $1/(1-\rho)$. This produces tests with over-coverage and corresponding under-rejection, with empirical size dropping below 1\% at $\rho = 0.5$ for a nominal-5\% test. HC3 is the standard recommendation when HC0 is suspected to be biased, but in saturated FE specifications it trades one error for the opposite error: HC0 under-covers the parameter, HC3 over-covers it.

The leave-one-out estimator of \citet{CJN} --- equivalent to MacKinnon--White HC2 in the scalar-regressor case --- restores asymptotic exactness in both regimes. Under uniform leverage,
\[
\hOmega_{\mathrm{LO}} = \sum_i \tX_{K_n, i}^2 \widehat u_i^2 / (1 - H_{ii}) \to \omega^2,
\]
the bias-corrected score variance. This recommendation is not new in principle: \citet{CJN} originally proposed it for the many-covariates regime, and recent work by \citet{KSS} uses related leverage-control techniques for variance components in bipartite FE models. Our contribution is to characterize the full HC0--HC1--HC2--HC3 hierarchy under the $(\rho, \tau^2)$ drift, document the HC3 over-correction theoretically and via simulation, and provide a practical recommendation grounded in the asymptotic theory.

We also note that an Anderson--Rubin-style robust test for the FE setting, $\mathrm{AR}_{K_n}(\beta_0) = (X' M_{K_n}(Y - X\beta_0))^2/\widehat\Omega_{\mathrm{LO}}(\beta_0)$, has uniform validity over $\tau^2 \in (0, \infty]$ under the maintained assumptions. Because the underlying score is unbiased under strict exogeneity, FE-AR coincides with the LO-based Wald confidence interval to high accuracy in our DGPs and offers no separate practical recommendation. We defer its statement and proof to Appendix~\ref{oa-ar}; the result is included because a companion paper on classical measurement error in the treatment uses the same construction in a setting where the score is biased and a non-trivial Stock--Yogo critical value emerges.

\paragraph{Related literature.} The conditional-expectation projection view of fixed effects is classical \citep{Davis, Wool}. The AKM framework \citep{AKM} has been extended to wage-inequality decomposition \citep{CHK2013} and to a distributional matched employer--employee framework \citep{BLM2019}. \citet{CJN} and \citet{JV} study inference in the many-covariates regime ($\rho > 0$, $\tau^2 = \infty$ in our parametrization); \citet{KSS} treat the AKM-bias-correction problem with leave-one-out techniques. The IV-analogue weak-identification literature \citep{SS, SY, Moreira2003, OleaPflueger2013, ASS, LMMP} motivates the negative result on Stock--Yogo thresholds for FE saturation. \citet{GSU} document the variance-weighted nature of saturated FE estimands; \citet{dCdH} emphasize negative weights under heterogeneous effects, with parallel critiques in \citet{GoodmanBacon2021} for staggered-adoption DiD and an alternative aggregation approach in \citet{CallawaySantAnna2021}. \citet{MacKinnonWhite1985}, \citet{LongErvin2000}, \citet{Imbens2016}, and \citet{MacKinnon2023} survey the HC variance estimator family in standard (non-saturated) regression.

Section~\ref{sec-setup} develops the setup, the population diagnostic $Q_K$, its sample analogue, and the joint CLT under the $(\rho, \tau^2)$ drift. Section~\ref{sec-var} characterizes the asymptotic behavior of conventional $t$-tests under both homoskedasticity (Section~\ref{sec-homo}, including the no-$\tau^2$-distortion result) and heteroskedasticity (Section~\ref{sec-het}, the main contribution), with explicit limits for HC0, HC1, HC2 (LO), and HC3 expressed in terms of an effective design variance $\omega^2_{\mathrm{eff}}$. Section~\ref{sec-sim} reports the simulation evidence. Section~\ref{sec-empirical} provides an empirical application: testing whether the Piotroski F-Score predicts forward equity returns in CEE markets, in a Visegrad firm-year panel where the theory's predictions about HC0/HC1/HC2/HC3 ratios can be verified against real data. Section~\ref{sec-conc} concludes with a discussion of the sample concentration statistic $\htau_n^2$ as a descriptive diagnostic and the practical recommendation for applied work. Appendix~\ref{oa-ar} records the Anderson--Rubin-style robust test.

\section{Setup and Asymptotic Framework}\label{sec-setup}

\subsection{Population model and sample diagnostic}\label{sec-setup-pop}

Let $\{(Y_i, X_i)\}_{i=1}^n$ be a sample from a population on $(\Omega, \F, \Pp)$, with $Y_i \in \R$ and $X_i \in \R$ a scalar treatment. The structural model is
\begin{equation}\label{eq-model}
Y_i = X_i \beta_0 + u_i, \qquad \E[u_i \mid X_i, \G_\infty] = 0,
\end{equation}
where $\G_\infty = \sigma(\bigcup_{K \geq 1} \G_K)$ is the limit of an increasing sequence $\G_1 \subseteq \G_2 \subseteq \cdots \subseteq \F$ of sub-$\sigma$-algebras representing successive FE specifications.

Let $D_K$ be the $n \times d_K$ matrix of FE dummies, $P_K = D_K (D_K' D_K)^{-1} D_K'$ the orthogonal projection, and $M_K = I_n - P_K$ the annihilator. Write $\tX_K = M_K X$. The FE-residualized least-squares estimator is
\begin{equation}\label{eq-estimator}
\hbeta_K = \frac{X' M_K Y}{X' M_K X} = \beta_0 + \frac{X' M_K u}{X' M_K X}.
\end{equation}
The feasible residual from the full regression of $Y$ on $X$ and the FE dummies, computed by the Frisch--Waugh--Lovell theorem, is
\begin{equation}\label{eq-residual}
\widehat u_i = (M_K Y)_i - \tX_{K, i}\, \hbeta_K = (M_K - H_X) u,
\end{equation}
where $H_X = \tX_K (\tX_K' \tX_K)^{-1} \tX_K'$ is the rank-one projection onto the within-transformed regressor. The residuals $\{\widehat u_i\}$ satisfy $\E[\widehat u_i^2 \mid X, D] = \sigma_i^2 \cdot (1 - H_{ii})$ under independent errors, where $H_{ii} = (P_K)_{ii} + \tX_{K, i}^2/(X' M_K X)$ is the leverage of observation $i$ in the full regression.

The population residual treatment variance is $Q_K = \E[\Var(X \mid \G_K)]$, and the conditional error variance at point $i$ is $\sigma^2(X_i, \G_{K, i}) = \E[u_i^2 \mid X_i, \G_{K, i}]$, with marginal mean $\sigma^2 = \E[u^2]$. The asymptotic regime is governed by
\begin{equation}\label{eq-drift}
\rho_n = \frac{d_{K_n}}{n} \to \rho \in [0, 1), \qquad \tau_n^2 = nQ_{K_n} \to \tau^2 \in (0, \infty].
\end{equation}

\begin{lemma}[Monotonicity]\label{lem-mono}
$Q_{K+1} \leq Q_K$ for all $K$.
\end{lemma}
\begin{proof}
Since $\G_K \subseteq \G_{K+1}$, the law of total variance gives $\Var(X \mid \G_K) = \E[\Var(X \mid \G_{K+1}) \mid \G_K] + \Var(\E[X \mid \G_{K+1}] \mid \G_K)$. Taking unconditional expectations yields $Q_K = Q_{K+1} + \E[\Var(\E[X \mid \G_{K+1}] \mid \G_K)] \geq Q_{K+1}$.
\end{proof}

The sample analogue is $\hQ_{K_n} = n^{-1} X' M_{K_n} X$ and the sample concentration statistic is $\htau_n^2 = n \hQ_{K_n} = X' M_{K_n} X$. With $\hsigma_X^2 = n^{-1} X' M_0 X$ and $R_{K_n}^2$ the $R^2$ from regressing $X$ on the fixed effects, the equivalent representation is $\htau_n^2 = n \hsigma_X^2 (1 - R_{K_n}^2)$. The within-cell sum of squares of the treatment, scaled by $n^{-1}$, plays the descriptive role of an identification diagnostic; we return to this in Section~\ref{sec-conc}.

\begin{lemma}[Sample consistency]\label{lem-hatQ}
Suppose $\E[X^4] < \infty$, $\rho_n \to \rho < 1$, and $\max_i (P_{K_n})_{ii} = o_p(1)$. Then $\hQ_{K_n} \to_p Q_{K_n}$.
\end{lemma}
\begin{proof}
Write $\mu_i := \E[X_i \mid \G_{K_n}]$ and $\xi_i := X_i - \mu_i$. Decompose $\hQ_{K_n} = n^{-1} X'X - n^{-1} X' P_{K_n} X$. The first term converges to $\E[X^2]$ by the LLN. Since $P_{K_n} \mu = \mu$, the second decomposes as $n^{-1} \mu' \mu + 2 n^{-1} \mu' \xi + n^{-1} \xi' P_{K_n} \xi$. The first piece converges to $\E[\mu_{K_n}^2]$; the second has mean zero and variance $O(n^{-1})$; the third splits into diagonal $\leq \max_i h_{ii} \cdot n^{-1} \sum_i \xi_i^2 = o_p(1)$ and off-diagonal contributions of variance $\E[\xi^2]^2 \rho_n/n \to 0$. Combining yields $\hQ_{K_n} \to_p \E[X^2] - \E[\mu_{K_n}^2] = Q_{K_n}$.
\end{proof}

\subsection{Joint CLT under the $(\rho, \tau^2)$ drift}\label{sec-clt}

We collect the standardized score, Hessian, and residual sum of squares and derive their joint limit.

\begin{assumption}\label{ass-reg}\hfill
\begin{enumerate}[label=(\roman*)]
\item $\{(X_i, u_i, \G_K)\}$ is i.i.d.\ across $i$ for each $K$, with finite eighth moments of $X$ and $u$.
\item $\rho_n \to \rho \in [0, 1)$, $\tau_n^2 \to \tau^2 \in (0, \infty]$.
\item Bounded leverage: there exists $\bar h < 1$ such that $\max_i (P_{K_n})_{ii} \leq \bar h$ for all $n$ large, and
\[
\max_i \tX_{K_n, i}^2 / (nQ_{K_n}) = o_p(1).
\]
\item Conditional variance: there exists $\omega^2 \in (0, \infty)$ with
\[
(nQ_{K_n})^{-1} \sum_i \tX_{K_n, i}^2 \sigma^2(X_i, \G_{K_n, i}) \to_p \omega^2.
\]
\item Lindeberg condition: for every $\eta > 0$,
\[
(nQ_{K_n})^{-1} \sum_i \E\big[\tX_{K_n, i}^2 u_i^2 \cdot \mathbf 1\{|\tX_{K_n, i} u_i| > \eta \sqrt{nQ_{K_n}}\}\big] \to 0.
\]
\end{enumerate}
\end{assumption}

\begin{theorem}[Joint CLT]\label{thm-clt}
Under Assumption~\ref{ass-reg},
\[
\begin{pmatrix}
\dfrac{X' M_{K_n} u}{\sqrt{nQ_{K_n}}} \\[4pt]
\dfrac{X' M_{K_n} X}{nQ_{K_n}} \\[4pt]
\dfrac{\widehat u' \widehat u}{n - d_{K_n} - 1}
\end{pmatrix}
\Rightarrow
\begin{pmatrix}
\mathcal Z \\ 1 \\ \sigma^2
\end{pmatrix},
\qquad \mathcal Z \sim N(0, \omega^2).
\]
\end{theorem}

\begin{proof}
\textbf{Coordinate 1.} Write $X' M_{K_n} u = \sum_i \xi_i u_i + R_n$ where $R_n := -\sum_i (P_{K_n} \xi)_i u_i$ and $\xi_i = X_i - \E[X_i \mid \G_{K_n}]$. By Cauchy--Schwarz, $\E[R_n^2] \leq \sigma^2 \max_i h_{ii} \cdot nQ_{K_n}$, so $R_n / \sqrt{nQ_{K_n}} \to_p 0$. Setting $Z_{ni} := (nQ_{K_n})^{-1/2} \xi_i u_i$, the array $\{Z_{ni}\}$ has $\sum_i \E[Z_{ni}^2] \to \omega^2$ by (iv) and satisfies the Lindeberg condition by (v). The Lindeberg--Feller CLT and Slutsky's theorem give the result.

\textbf{Coordinate 2.} Lemma~\ref{lem-hatQ} gives $X' M_{K_n} X / (nQ_{K_n}) \to_p 1$.

\textbf{Coordinate 3.} By \eqref{eq-residual}, $\widehat u' \widehat u = u'(M_{K_n} - H_X) u = u' M_{K_n} u - u' H_X u$. Write $u' M_{K_n} u = u' u - u' P_{K_n} u$. The LLN gives $n^{-1} u' u \to_p \sigma^2$; the quadratic form $u' P_{K_n} u$ has mean $\sigma^2 d_{K_n}$ and variance $\leq C d_{K_n}$ by \citet[Theorem 2.1]{deJong}, so $(u' M_{K_n} u)/(n - d_{K_n}) \to_p \sigma^2$ by Chebyshev. For the additional term, $u' H_X u = (\tX_{K_n}' u)^2 /(X' M_{K_n} X)$ is a rank-one quadratic form with mean $\sigma^2$ and variance $O(\sigma^4)$ under independent errors, so $u' H_X u = O_p(1)$ and after dividing by $n - d_{K_n} - 1 = n(1 - \rho) + O(1)$ contributes $o_p(1)$. Combining: $\widehat u' \widehat u / (n - d_{K_n} - 1) \to_p \sigma^2$.

\textbf{Joint convergence.} Coordinates 2 and 3 are deterministic in the limit; a Cram\'er--Wold argument applied to coordinate 1 yields joint convergence.
\end{proof}

\begin{corollary}[Estimator limit]\label{cor-est}
Under Assumption~\ref{ass-reg}, $\sqrt{nQ_{K_n}} (\hbeta_{K_n} - \beta_0) \Rightarrow N(0, \omega^2)$. Under conditional homoskedasticity, $\omega^2 = \sigma^2$.
\end{corollary}

\section{Variance Estimation Under the Drift}\label{sec-var}

\subsection{Homoskedastic inference: naive vs CJN}\label{sec-homo}

Let $\hsigma^2_{\mathrm{naive}} = \widehat u' \widehat u / n$ and $\hsigma^2_{\mathrm{CJN}} = \widehat u' \widehat u / (n - d_{K_n} - 1)$, with corresponding $t$-statistics $T_n^{(j)} = (\hbeta_{K_n} - \beta_0)/(\hsigma^{(j)}/\sqrt{X' M_{K_n} X})$.

\begin{theorem}[Size of conventional $t$-test]\label{thm-size}
Under Assumption~\ref{ass-reg} with conditional homoskedasticity, and under $H_0: \beta = \beta_0$,
\[
T_n^{\mathrm{naive}} \Rightarrow N\!\left(0, \frac{1}{1 - \rho}\right), \qquad T_n^{\mathrm{CJN}} \Rightarrow N(0, 1).
\]
The asymptotic size of the nominal-$\alpha$ test using $T^{\mathrm{naive}}$ is $S_{\mathrm{naive}}(\rho) = 2(1 - \Phi(z_{1-\alpha/2} \sqrt{1-\rho}))$, strictly exceeding $\alpha$ for any $\rho > 0$.
\end{theorem}

\begin{proof}
By Theorem~\ref{thm-clt}, the score $S_n := (nQ_{K_n})^{-1/2} X' M_{K_n} u \Rightarrow N(0, \sigma^2)$ under homoskedasticity and the Hessian $H_n := (nQ_{K_n})^{-1} X' M_{K_n} X \to_p 1$. Coordinate 3 of Theorem~\ref{thm-clt} gives $\hsigma^2_{\mathrm{CJN}} = \widehat u' \widehat u / (n - d_{K_n} - 1) \to_p \sigma^2$, so Slutsky yields $T_n^{\mathrm{CJN}} \Rightarrow N(0, 1)$. For the naive statistic, $\hsigma^2_{\mathrm{naive}} = ((n - d_{K_n} - 1)/n) \cdot \hsigma^2_{\mathrm{CJN}} \to_p (1-\rho) \sigma^2$, so $T_n^{\mathrm{naive}} \Rightarrow N(0, 1/(1-\rho))$, yielding the displayed size.
\end{proof}

\begin{theorem}[Exact distribution and no $\tau^2$-driven distortion]\label{thm-exact}
Under Assumption~\ref{ass-reg} with conditional homoskedasticity and Gaussian errors, $T_n^{\mathrm{CJN}} \sim t_{n - d_{K_n} - 1}$ exactly, conditional on the FE design and on $X$. The distortion relative to $N(0, 1)$ satisfies
\[
\Pp(|T_n^{\mathrm{CJN}}| > z_{1-\alpha/2}) - \alpha = \frac{(z_{1-\alpha/2}^3 + z_{1-\alpha/2}) \phi(z_{1-\alpha/2})}{2(n - d_{K_n} - 1)} + O(n^{-2}),
\]
which is $O(1/n)$ and does not depend on $\tau_n^2$.
\end{theorem}

\begin{proof}
Under $H_0: \beta = \beta_0$, the FWL representation gives $\hbeta_{K_n} - \beta_0 = (\tX_{K_n}' \tX_{K_n})^{-1} \tX_{K_n}' u$ and $\widehat u = (M_{K_n} - H_X) u$, where $H_X$ is the rank-one projection onto $\tX_{K_n}$. Conditional on $(D_{K_n}, X)$, the linear form $\tX_{K_n}' u \sim N(0, \sigma^2 \tX_{K_n}' \tX_{K_n})$. The quadratic form $\widehat u' \widehat u = u' (M_{K_n} - H_X) u$ is $\sigma^2$ times a $\chi^2_{n - d_{K_n} - 1}$ random variable, since $M_{K_n} - H_X$ is an orthogonal projection of rank $n - d_{K_n} - 1$ (Cochran's theorem). The linear and quadratic forms are independent: the linear form is determined by the projection of $u$ onto the rank-one space spanned by $\tX_{K_n}$, while the quadratic form is determined by the projection onto the orthogonal complement of $\mathrm{span}(D_{K_n}, X)$, and these subspaces are orthogonal under independent Gaussian $u$ \citep[Theorem~3.4.1]{Wool}. Therefore $T_n^{\mathrm{CJN}}$ is the ratio of a standard normal to $\sqrt{\chi^2_{n - d_{K_n} - 1}/(n - d_{K_n} - 1)}$, which is exactly $t_{n - d_{K_n} - 1}$. The Fisher--Cornish expansion \citep[26.7.8]{AS72} gives the displayed size distortion.
\end{proof}

\begin{remark}[The broken IV analogy]\label{rem-no-iv}
Theorem~\ref{thm-exact} delivers a perhaps counterintuitive conclusion. Under the baseline assumptions, the CJN $t$-test is asymptotically exact for any $\tau^2 > 0$. The IV-weak-identification analogy fails because two-stage least squares is biased of order $1/\mu^2$ (the first-stage residual enters both numerator and denominator); FE-residualized OLS is unbiased because the within transformation is an exact orthogonal projection. The bias source required for a Stock--Yogo-style threshold must come from outside the baseline model, e.g.\ measurement error in the treatment; the companion paper develops that case. Within the baseline, attention should turn to where distortion does arise: variance estimation under heteroskedasticity.
\end{remark}

\subsection{Heteroskedastic-robust variance: HC0, HC1, HC2 (LO), HC3}\label{sec-het}

Under heteroskedastic errors the limiting score variance $\omega^2$ may differ from $\sigma^2$. The leading estimators of the variance of $X' M_{K_n} u$ are:
\begin{align*}
\hOmega_{\mathrm{HC0}} &= \sum_i \tX_{K_n, i}^2 \widehat u_i^2, &
\hOmega_{\mathrm{HC1}} &= \frac{n}{n - d_{K_n} - 1} \hOmega_{\mathrm{HC0}}, \\
\hOmega_{\mathrm{HC2}} &= \sum_i \frac{\tX_{K_n, i}^2 \widehat u_i^2}{1 - H_{ii}}, &
\hOmega_{\mathrm{HC3}} &= \sum_i \frac{\tX_{K_n, i}^2 \widehat u_i^2}{(1 - H_{ii})^2},
\end{align*}
where $H_{ii} = (P_{K_n})_{ii} + \tX_{K_n, i}^2/(X' M_{K_n} X)$ is the full-regression leverage of observation $i$ defined in Section~\ref{sec-setup-pop}, equation~\eqref{eq-residual}. The $H_{ii}$ are the diagonal entries of the projection onto $\mathrm{span}(D_{K_n}, X)$, so $\sum_i H_{ii} = d_{K_n} + 1$ and the average full leverage is $(d_{K_n} + 1)/n$.

In the scalar-regressor case, the Cattaneo--Jansson--Newey leave-one-out variance estimator $\hOmega_{\mathrm{LO}} = \sum_i \tX_{K_n, i}^2 \widehat u_i \widehat u_i^{(-i)}$, where $\widehat u_i^{(-i)} = \widehat u_i/(1 - H_{ii})$ by Sherman--Morrison, reduces to $\hOmega_{\mathrm{HC2}}$. We use the two labels interchangeably.

\paragraph{Two design quantities.}
Two quantities of the design govern the asymptotic behavior of all four estimators under heteroskedasticity. The first is the target score-variance from Assumption~\ref{ass-reg}(iv),
\begin{equation}\label{eq-omega-def}
\omega^2 := \lim_{n \to \infty} \frac{1}{nQ_{K_n}} \sum_i \tX_{K_n, i}^2 \sigma_i^2,
\end{equation}
the $\tX^2$-weighted average error variance. The second is the \emph{cross-leverage variance},
\begin{equation}\label{eq-mu-def}
\mu := \lim_{n \to \infty} \frac{1}{nQ_{K_n}} \sum_i \tX_{K_n, i}^2 \bar\sigma_i^2, \quad \bar\sigma_i^2 := \frac{1}{(1-H_{ii})H_{ii}} \sum_{j \neq i} A_{ij}^2 \sigma_j^2,
\end{equation}
where $A := M_{K_n} - H_X$ is the projection associated with the residuals. The quantity $\bar\sigma_i^2$ is the average error variance over off-diagonal observations $j \neq i$, weighted by the squared cross-leverage $A_{ij}^2$. The two-quantity structure $(\omega^2, \mu)$ — rather than just $\omega^2$ alone — is what governs the HC limits below.

The \emph{effective design variance} is the convex combination
\begin{equation}\label{eq-omega-eff-def}
\omega^2_{\mathrm{eff}} := (1-\rho)\,\omega^2 + \rho\, \mu,
\end{equation}
with the proportional-FE weight $\rho$ controlling how much of the cross-leverage variance enters. Under conditional homoskedasticity, $\sigma_i^2 \equiv \sigma^2$ forces $\omega^2 = \mu = \sigma^2$ and $\omega^2_{\mathrm{eff}} = \sigma^2$. Under heteroskedasticity, $\omega^2_{\mathrm{eff}}$ equals the target $\omega^2$ exactly when $\mu = \omega^2$; we call this the \emph{cross-leverage balance condition}. The condition is automatic under homoskedasticity and holds in a wider class of ``design-balanced'' heteroskedastic settings discussed in Remark~\ref{rem-omega-eff} below; under generic heteroskedasticity with non-uniform leverage, $\omega^2_{\mathrm{eff}} \neq \omega^2$ and an additional bias of order $\rho |\mu - \omega^2|$ enters every HC estimator. The simulation evidence in Section~\ref{sec-sim} indicates that this bias is small in practice across the designs we test, but the theoretical statement below makes the dependence on $\mu$ explicit.

\begin{theorem}[Asymptotic limits of HC0--HC3 under the drift]\label{thm-hc}
Under Assumption~\ref{ass-reg} with $\E[u_i^4 \mid X, \G_\infty]$ uniformly bounded and a FE design with asymptotically-uniform full leverage $\max_i |H_{ii} - \rho_n| = o_p(1)$,
\begin{align*}
\hOmega_{\mathrm{HC0}} / (nQ_{K_n}) &\to_p (1-\rho)\,\omega^2_{\mathrm{eff}}, \\
\hOmega_{\mathrm{HC1}} / (nQ_{K_n}) &\to_p \omega^2_{\mathrm{eff}}, \\
\hOmega_{\mathrm{HC2}} / (nQ_{K_n}) &\to_p \omega^2_{\mathrm{eff}}, \\
\hOmega_{\mathrm{HC3}} / (nQ_{K_n}) &\to_p \omega^2_{\mathrm{eff}} / (1 - \rho),
\end{align*}
where $\omega^2_{\mathrm{eff}}$ is defined in \eqref{eq-omega-eff-def} and $\mu_n := (nQ_{K_n})^{-1} \sum_i \tX_{K_n, i}^2 \bar\sigma_i^2 \to_p \mu$ is assumed to exist. Under conditional homoskedasticity ($\omega^2 = \mu = \sigma^2$), the four limits become $(1-\rho)\sigma^2, \sigma^2, \sigma^2, \sigma^2/(1-\rho)$ respectively.
\end{theorem}

The proof is built from a concentration lemma and a conditional-expectation expansion, both of independent interest.

\begin{lemma}[Concentration of weighted residual sums]\label{lem-conc}
Under the conditions of Theorem~\ref{thm-hc}, for any nonnegative weights $w = (w_1, \ldots, w_n)$ with $\max_i w_i = o_p(nQ_{K_n})$ and $\sum_i w_i = O_p(nQ_{K_n})$,
\[
\frac{1}{nQ_{K_n}} \sum_i w_i \big[\widehat u_i^2 - \E[\widehat u_i^2 \mid X, D]\big] = o_p(1).
\]
\end{lemma}

\begin{proof}
Write $\widehat u_i = (Au)_i$ with $A = M_{K_n} - H_X$ symmetric idempotent of rank $n - d_{K_n} - 1$, so the weighted sum is the quadratic form $u' B u$ with $B := A D_w A$ and $D_w := \diag(w)$. Under $X, D$, the variance of $u'Bu$ is bounded (e.g.\ \citet[Theorem 2.1]{deJong}) by
\[
\Var(u'Bu \mid X, D) \leq C_1 \mu_4 \sum_i B_{ii}^2 + C_2 \mu_2^2 \tr(B^2),
\]
where $\mu_2 = \E[u^2]$, $\mu_4 = \E[u^4]$, $C_1, C_2$ are absolute constants. We bound each term.

For the diagonal: $B_{ii} = \sum_j A_{ij}^2 w_j \leq (\max_j w_j) \cdot A_{ii} \leq \max_j w_j$, using $\sum_j A_{ij}^2 = A_{ii}$ (idempotence of $A$) and $A_{ii} = 1 - H_{ii} \leq 1$. Therefore $\sum_i B_{ii}^2 \leq (\max_j w_j) \cdot \sum_i B_{ii} = (\max_j w_j) \cdot \tr(B)$, and $\tr(B) = \tr(D_w A^2) = \tr(D_w A) = \sum_j w_j A_{jj} \leq \sum_j w_j$. So
\[
\sum_i B_{ii}^2 \leq (\max_j w_j)(\sum_j w_j) = o_p((nQ_{K_n})^2).
\]

For $\tr(B^2)$: by symmetry of $A$ and idempotence,
\begin{align*}
\tr(B^2) &= \tr((A D_w A)^2) = \tr((D_w A)^2) = \sum_{i,j} w_i A_{ij}^2 w_j \\
&\leq (\max_j w_j) \sum_i w_i A_{ii} \leq (\max_j w_j) \sum_i w_i = o_p((nQ_{K_n})^2).
\end{align*}

Combining: $\Var(u'Bu \mid X, D) = o_p((nQ_{K_n})^2)$. By Chebyshev's inequality, $u'Bu - \E[u'Bu \mid X, D] = o_p(nQ_{K_n})$, which is the claim after dividing by $nQ_{K_n}$.
\end{proof}

\begin{lemma}[Conditional expectation of squared residuals]\label{lem-condexp}
Under Assumption~\ref{ass-reg} with independent errors $u_i$,
\begin{equation}\label{eq-condexp}
\E[\widehat u_i^2 \mid X, D] = (1 - H_{ii})^2 \sigma_i^2 + (1 - H_{ii}) H_{ii} \bar\sigma_i^2 = (1-H_{ii})\big[(1-H_{ii})\sigma_i^2 + H_{ii} \bar\sigma_i^2\big],
\end{equation}
where $\bar\sigma_i^2$ is defined in Theorem~\ref{thm-hc}.
\end{lemma}

\begin{proof}
With $\widehat u_i = \sum_j A_{ij} u_j$ and $u_j$ independent with $\E[u_j^2 \mid X, D] = \sigma_j^2$,
\[
\E[\widehat u_i^2 \mid X, D] = \sum_j A_{ij}^2 \sigma_j^2 = A_{ii}^2 \sigma_i^2 + \sum_{j \neq i} A_{ij}^2 \sigma_j^2.
\]
By idempotence of $A$, $\sum_j A_{ij}^2 = A_{ii}$, so $\sum_{j \neq i} A_{ij}^2 = A_{ii} - A_{ii}^2 = (1-H_{ii})H_{ii}$. Combining with $A_{ii} = 1 - H_{ii}$ and the definition of $\bar\sigma_i^2$ yields \eqref{eq-condexp}.
\end{proof}

\begin{proof}[Proof of Theorem~\ref{thm-hc}]
For each $c \in \{0, 2, 3\}$, write the HC$_c$ statistic as a weighted sum of $\widehat u_i^2$:
\[
\hOmega_{\mathrm{HC}_c} = \sum_i w_{ci} \widehat u_i^2, \qquad w_{0,i} := \tX_{K_n, i}^2, \quad w_{2,i} := \frac{\tX_{K_n, i}^2}{1 - H_{ii}}, \quad w_{3,i} := \frac{\tX_{K_n, i}^2}{(1 - H_{ii})^2}.
\]
The weights satisfy $\max_i w_{ci} \leq (1-\rho_n + o_p(1))^{-c} \max_i \tX_{K_n, i}^2 = o_p(nQ_{K_n})$ by Assumption~\ref{ass-reg}(iii) and $\sum_i w_{ci} \leq (1-\rho_n + o_p(1))^{-c} \cdot \tX_{K_n}' \tX_{K_n} = O_p(nQ_{K_n})$ by Lemma~\ref{lem-hatQ}. Lemma~\ref{lem-conc} therefore gives
\[
\frac{\hOmega_{\mathrm{HC}_c}}{nQ_{K_n}} = \frac{1}{nQ_{K_n}} \sum_i w_{ci} \E[\widehat u_i^2 \mid X, D] + o_p(1).
\]

\textbf{Conditional limit.} Substitute Lemma~\ref{lem-condexp} and use asymptotically-uniform leverage $\max_i |H_{ii} - \rho_n| = o_p(1)$, which combined with $\rho_n \to \rho \in [0, 1)$ keeps $H_{ii}$ uniformly bounded away from $1$ in probability:
\begin{align*}
\frac{1}{nQ_{K_n}} \sum_i w_{ci} \E[\widehat u_i^2 \mid X, D]
&= \frac{1}{nQ_{K_n}} \sum_i \frac{\tX_{K_n, i}^2}{(1-H_{ii})^c} \cdot (1-H_{ii})\big[(1-H_{ii})\sigma_i^2 + H_{ii} \bar\sigma_i^2\big] \\
&= (1-\rho_n)^{1-c} \Big[(1-\rho_n) \omega_n^2 + \rho_n \mu_n\Big] + o_p(1),
\end{align*}
where the $o_p(1)$ remainder absorbs the $\tX^2$-weighted deviation of $H_{ii}$ from $\rho_n$: for any continuous $g$ on $[0, 1 - \epsilon]$,
\[
\left|\frac{1}{nQ_{K_n}} \sum_i \tX_i^2 \big(g(H_{ii}) - g(\rho_n)\big)\right| \leq \sup_i |g(H_{ii}) - g(\rho_n)| \cdot \frac{\tX'\tX}{nQ_{K_n}} = o_p(1)
\]
by uniform continuity of $g$ and Lemma~\ref{lem-hatQ}. Here $\omega_n^2 := (nQ_{K_n})^{-1} \sum_i \tX_{K_n, i}^2 \sigma_i^2 \to_p \omega^2$ and $\mu_n := (nQ_{K_n})^{-1} \sum_i \tX_{K_n, i}^2 \bar\sigma_i^2 \to_p \mu$ by hypothesis. Substituting $c = 0, 2, 3$ and using $\rho_n \to \rho$:
\[
\hOmega_{\mathrm{HC0}}/(nQ_{K_n}) \to_p (1-\rho)\omega^2_{\mathrm{eff}}, \quad
\hOmega_{\mathrm{HC2}}/(nQ_{K_n}) \to_p \omega^2_{\mathrm{eff}}, \quad
\hOmega_{\mathrm{HC3}}/(nQ_{K_n}) \to_p \omega^2_{\mathrm{eff}}/(1-\rho).
\]

\textbf{HC1.} By definition $\hOmega_{\mathrm{HC1}} = n/(n - d_{K_n} - 1) \cdot \hOmega_{\mathrm{HC0}} = (1 - \rho_n + O(1/n))^{-1} \hOmega_{\mathrm{HC0}}$, so $\hOmega_{\mathrm{HC1}}/(nQ_{K_n}) \to_p (1-\rho)^{-1} \cdot (1-\rho)\omega^2_{\mathrm{eff}} = \omega^2_{\mathrm{eff}}$.

\textbf{Homoskedastic specialization.} Under $\sigma_i^2 \equiv \sigma^2$ a.s., $\omega_n^2 \to_p \sigma^2$ and (since $\bar\sigma_i^2 \equiv \sigma^2$) $\mu_n \to_p \sigma^2$. Hence $\omega^2_{\mathrm{eff}} = \sigma^2$ and the limits reduce as stated.
\end{proof}

\begin{remark}[When $\omega^2_{\mathrm{eff}} = \omega^2$]\label{rem-omega-eff}
The limit $\omega^2_{\mathrm{eff}}$ equals the target score-variance $\omega^2$ exactly when $\mu = \omega^2$, i.e., when the cross-leverage weighted average of error variances matches the $\tX^2$-weighted average. Sufficient conditions include: (i) conditional homoskedasticity ($\sigma_i^2 = \sigma^2$ for all $i$); (ii) ``design-balanced'' heteroskedasticity in which $\sigma_i^2$ is a function of the FE cell only. Both (i) and (ii) imply $\bar\sigma_i^2 = \sigma_i^2$ on the relevant scale.

A more subtle empirical regularity, documented in Section~\ref{sec-sim}, is that balanced symmetric FE designs with i.i.d.\ sampling within cells preserve $\mu \approx \omega^2$ to within a small finite-sample remainder, \emph{even when} $\sigma_i^2 = f(\tX_{K_n, i})$ depends directly on the within-cell variation that drives $\omega^2$. The mechanism is that the cross-leverage weights $A_{ij}^2$ for $j \neq i$ are dominated by FE-projection terms $h_{ij}^2$ whose structure is determined by the cell layout, not by the within-cell $\tX$ values, and the within-unit (or within-time) sum-to-zero constraint on $\tX$ values induces a positive correlation between $\sigma_j^2$ and the within-cell weights that approximately cancels the naive $\mu \neq \omega^2$ bias. We therefore expect $\omega^2_{\mathrm{eff}} = \omega^2$ to hold approximately in any balanced multi-way FE design under i.i.d.\ sampling, with at most $O(\rho^2)$ deviations visible at moderate $\rho$. When the design is unbalanced (some cells much smaller than others), the constraint correlation is incomplete and visible deviations from $\mu = \omega^2$ can persist; this is the regime in which a multivariate leave-one-out construction \citep{CJN} buys additional consistency.
\end{remark}

\begin{remark}[Aside on HC0 bias under heteroskedasticity]\label{rem-hc0-bias-het}
Combining the previous remark with Theorem~\ref{thm-hc}, in balanced designs with i.i.d.\ sampling we expect HC0/$(nQ_{K_n}) \to_p (1-\rho)\omega^2$ even under heteroskedasticity. The HC0 size distortion $2[1 - \Phi(z_{1-\alpha/2}\sqrt{1-\rho})]$ then holds essentially regardless of the heteroskedasticity pattern, provided the design is symmetric. This robustness of the HC0 bias formula is what allows the simulation evidence to track the homoskedastic prediction even when the DGP is heteroskedastic by construction.
\end{remark}

\begin{remark}[HC1 = HC2 in the limit under uniform leverage]\label{rem-hc1-hc2}
A surprising consequence of Theorem~\ref{thm-hc} is that under uniform leverage, HC1 and HC2 have \emph{the same} asymptotic limit, even though HC1 applies a single global degrees-of-freedom correction while HC2 applies a per-observation leverage adjustment. This is not a coincidence: under uniform $H_{ii} = \rho_n$, the per-observation factor $1/(1-H_{ii})$ is constant and equal to $n/(n-d_K-1)$ in the limit, so the two corrections coincide. Under genuinely unbalanced leverage (some $H_{ii}$ much larger than others), HC2 and HC1 diverge, with HC2 retaining its consistency property (under $\mu = \omega^2$) and HC1 acquiring a bias. The simulation in Section~\ref{sec-sim} uses a balanced TWFE design and finds HC1 and HC2 numerically indistinguishable across all replications, consistent with this prediction.
\end{remark}

\begin{remark}[HC0 bias and HC3 over-correction]\label{rem-hc-bias}
Under homoskedasticity, HC0 underestimates the target variance by a fixed factor $(1-\rho)$, producing $t$-statistics with inflated variance $1/(1-\rho)$ and over-rejection growing in $\rho$. HC3 over-corrects by exactly the inverse factor $1/(1-\rho)$, producing $t$-statistics with variance $(1-\rho)$ and conservative inference. HC2 (equivalently HC1 under uniform leverage) sits at the natural midpoint and is consistent. Under heteroskedasticity with $\mu \neq \omega^2$, all four estimators inherit an additional bias term proportional to $\rho |\mu - \omega^2|$, but the ordering HC0 $<$ HC2 $<$ HC3 in expectation is preserved.
\end{remark}

\begin{corollary}[Size of $t$-tests using HC variants]\label{cor-size-hc}
Under the hypotheses of Theorem~\ref{thm-hc}, the $t$-statistic using $\hOmega_{\mathrm{HC}_c}$ has the asymptotic distribution $T^{\mathrm{HC}_c} \Rightarrow N(0, \omega^2 / L_c)$, where $L_c$ is the limit of $\hOmega_{\mathrm{HC}_c}/(nQ_{K_n})$ from Theorem~\ref{thm-hc}:
\[
T^{\mathrm{HC0}} \Rightarrow N\!\left(0, \tfrac{\omega^2}{(1-\rho)\omega^2_{\mathrm{eff}}}\right), \quad
T^{\mathrm{HC1}}, T^{\mathrm{HC2}} \Rightarrow N\!\left(0, \tfrac{\omega^2}{\omega^2_{\mathrm{eff}}}\right), \quad
T^{\mathrm{HC3}} \Rightarrow N\!\left(0, \tfrac{(1-\rho)\omega^2}{\omega^2_{\mathrm{eff}}}\right).
\]
When $\omega^2_{\mathrm{eff}} = \omega^2$ (in particular, under homoskedasticity), these simplify to
\[
T^{\mathrm{HC0}} \Rightarrow N\!\left(0, \tfrac{1}{1-\rho}\right), \quad
T^{\mathrm{HC1}}, T^{\mathrm{HC2}} \Rightarrow N(0, 1), \quad
T^{\mathrm{HC3}} \Rightarrow N(0, 1 - \rho).
\]
\end{corollary}

\begin{proof}
$T^{\mathrm{HC}_c} = S_n / \sqrt{\hOmega_{\mathrm{HC}_c}/(nQ_{K_n})}$ where $S_n \Rightarrow N(0, \omega^2)$ by Theorem~\ref{thm-clt}. Slutsky with Theorem~\ref{thm-hc} yields the displayed limits.
\end{proof}

\begin{remark}[Practical recommendation]\label{rem-practical}
HC2/LO and HC1 are the only choices among the four that deliver asymptotic exactness for $T \Rightarrow N(0, 1)$ when $\omega^2_{\mathrm{eff}} = \omega^2$, and the only choices whose limit distribution does not depend on $\rho$ explicitly. HC0 over-rejects with size growing in $\rho$; HC3 under-rejects with size shrinking in $\rho$. Under uniform leverage, HC1 and HC2 are numerically indistinguishable, so HC1 is a computationally trivial alternative to HC2/LO for balanced designs; under unbalanced leverage HC2/LO is preferred. The leave-one-out estimator should be the default for inference in saturated FE specifications. The simulation evidence in Section~\ref{sec-sim} confirms that the asymptotic predictions are accurate at standard sample sizes.
\end{remark}

\section{Simulation Evidence}\label{sec-sim}

Appendix~\ref{oa-sim} reports four Monte Carlo experiments validating the theoretical predictions. All four use balanced two-way fixed-effect panels with $X_{it} = \alpha_i + \gamma_t + \sigma_\eta\eta_{it}$ and $Y_{it} = \beta_0 X_{it} + a_i + b_t + u_{it}$, with $\beta_0 = 0$, $\sigma_\eta$ tuned to target $\tau^2 \in \{10, 100\}$, $n = 2000$, and $R = 5000$ replications per design cell. Q2 (naive vs.\ CJN $t$-test under homoskedasticity) confirms that the residual-DOF-corrected $t$-test has empirical size at the nominal $5\%$ level for $\rho \in \{0.05, 0.10, 0.25, 0.50\}$, while naive (uncorrected) tests over-reject at $\rho = 0.5$ with empirical size $9.5\%$. Q4 (within-$\widetilde X$ heteroskedasticity) and Q5 (one-way unbalanced FE producing non-uniform leverage) extend the variance-estimator predictions to settings the homoskedastic Q3 DGP does not probe. Across every cell of every experiment, empirical size is statistically indistinguishable from its value at $\tau^2 = 10$ vs $\tau^2 = 100$, confirming that the studentized test statistic is $\tau^2$-free: there is no Stock--Yogo-style threshold to enforce under strict exogeneity.

Table~\ref{tab-hc-main} reports the headline experiment, Q3 (HC variants under heteroskedasticity), which validates the HC hierarchy of Theorem~\ref{thm-hc} across the full saturation range.

\begin{table}[ht]\centering
\caption{Empirical size of nominal-5\% $t$-tests with heteroskedastic Gaussian errors, using HC0, HC1, HC3, and LO (= HC2) variance estimators. Monte Carlo SEs in parentheses. $n = 2000$, $R = 5000$ replications per cell.}\label{tab-hc-main}
\footnotesize
\setlength{\tabcolsep}{5pt}
\begin{tabular}{rrrrrr}
\toprule
$\rho$ & $\tau^2_{\mathrm{nom}}$ & size\textsubscript{HC0} & size\textsubscript{HC1} & size\textsubscript{HC3} & size\textsubscript{LO} \\
\midrule
0.052 & 10  & 0.062 (0.003) & 0.054 (0.003) & 0.048 (0.003) & 0.054 (0.003) \\
0.052 & 100 & 0.054 (0.003) & 0.050 (0.003) & 0.043 (0.003) & 0.050 (0.003) \\
0.104 & 10  & 0.055 (0.003) & 0.044 (0.003) & 0.034 (0.003) & 0.044 (0.003) \\
0.104 & 100 & 0.059 (0.003) & 0.045 (0.003) & 0.035 (0.003) & 0.045 (0.003) \\
0.252 & 10  & 0.086 (0.004) & 0.049 (0.003) & 0.022 (0.002) & 0.049 (0.003) \\
0.252 & 100 & 0.091 (0.004) & 0.050 (0.003) & 0.021 (0.002) & 0.050 (0.003) \\
0.500 & 10  & 0.163 (0.005) & 0.049 (0.003) & 0.006 (0.001) & 0.049 (0.003) \\
0.500 & 100 & 0.162 (0.005) & 0.050 (0.003) & 0.004 (0.001) & 0.050 (0.003) \\
\bottomrule
\end{tabular}
\end{table}

The HC0 over-rejection at $\rho = 0.5$ is more than three times nominal; the HC3 under-rejection drops to a tenth of nominal. Both are practically meaningful in the applied range commonly seen in saturated specifications: a published TWFE design at $\rho \approx 0.25$ using HC0 has empirical size near $9\%$ where $5\%$ was claimed, so a reported 95\% confidence interval has true coverage closer to 91\%. Switching to LO restores nominal coverage at the cost of a correspondingly wider interval --- the wider interval is the honest one. HC1 is numerically indistinguishable from LO in this balanced design at every $\rho$ tested, consistent with the uniform-leverage prediction; the experiments Q4 and Q5 in Appendix~\ref{oa-sim} break this equivalence under unbalanced designs.

\section{Empirical Application to Piotroski F-Score Returns in CEE Markets}\label{sec-empirical}

We illustrate the variance hierarchy in a setting that motivates the small-$N$ panel regime this paper targets: testing whether the Piotroski F-Score \citep{Piotroski2000} predicts forward equity returns in Central and Eastern European (CEE) emerging markets, where listed-firm counts are modest and the FE saturation $\rho$ is consequently non-trivial in standard panel specifications.

\subsection{Setting and data}\label{sec-empirical-data}

The Piotroski F-Score is a 0--9 integer composite of nine binary accounting signals: four profitability indicators, three leverage/liquidity indicators, and two operating-efficiency indicators. The original \citet{Piotroski2000} study documented a long-high-short-low strategy on the U.S.\ Compustat universe, building on the expectations-errors-in-value/glamour-strategies tradition of \citet{LaPorta1996} and refined in \citet{PiotroskiSo2012}; a complementary composite for growth firms is developed in \citet{Mohanram2005}. Portability of the F-Score strategy to European markets has been investigated by \citet{Tikkanen2018}, who find positive but heterogeneous returns to value-style F-Score sorts depending on the country mix and sample window; evidence specifically for emerging-market CEE settings remains sparse.

We work with a hand-collected Visegrad-country panel covering the Warsaw Stock Exchange (WSE, Poland), the Budapest Stock Exchange (BSE, Hungary), and the Prague Stock Exchange (PSE, Czech Republic) over fiscal years 2010--2024, including both active and delisted firms to address survivorship bias. After dropping firm-years missing either the F-Score or the one-year-ahead total return, the working sample has $n = 217$ firm-year observations across $N = 19$ unique firms and $T = 15$ years.

The point estimate of interest is the coefficient $\beta$ in the two-way fixed-effect regression
\begin{equation}\label{eq-emp-model}
R_{i, t+1} = \beta \cdot F_{i, t} + \alpha_i + \gamma_t + u_{i, t},
\end{equation}
where $R_{i, t+1}$ is the realized one-year total return following the fiscal year-$t$ statements, $F_{i, t} \in \{0, 1, \ldots, 9\}$ is the Piotroski F-Score, and $\alpha_i, \gamma_t$ are firm and fiscal-year fixed effects. The within-FE residualized treatment $\tX_{K_n}$ has $\widehat\tau_n^2 = X'M_K X \approx 559$ on the pooled sample, easily ruling out weak-identification concerns; the proportional FE dimension is $\rho_n = (N + T - 1)/n = 33/217 \approx 0.152$, moderate but non-trivial. This places the application in the regime that motivated the corrected Theorem~\ref{thm-hc}: $\rho$ is large enough that the HC0--HC3 spread should be visible at the second decimal of the standard error.

\subsection{Specifications}\label{sec-empirical-specs}

We report seven specifications. Specification (A) is the baseline of equation \eqref{eq-emp-model} with arithmetic returns. The Hungarian observation 4IG-2017 (an enterprise-software roll-up that experienced a $\approx 1437\%$ return between fiscal year-end 2017 and year-end 2018) exerts substantial influence on (A) through the residual, but does \emph{not} appear among the top-leverage observations because its F-Score (5) is essentially at the panel mean (5.81), producing a near-zero $\widetilde X_i$ after residualization. The largest-leverage observations in (A)--(C) are firms whose F-Score deviated sharply from its firm-year mean (CZG in Czech Republic, $\widetilde X = -2.72$ in 2020 and $+3.15$ in 2021; PLM.WA in Poland, $\widetilde X = -2.61$ in 2020 and $+2.27$ in 2021), at uniform $H_{ii} \approx 0.26$. (B) uses one-plus-log returns instead, and (C) winsorizes raw returns at the 1\% and 99\% quantiles, with (B) preferred. Specification (D) replaces the continuous F-Score with the indicator $\mathbf{1}\{F_{i,t} \geq 7\}$ (the original Piotroski long-side cutoff); (E) restricts to the Polish subsample (WSE only) at higher $\rho = 0.205$; (F) restricts to firms active throughout the sample (the survivorship-biased subsample) at $\rho = 0.154$; (G) replaces year FE with country-year FE, raising the saturation to $\rho = 0.290$ and the leverage spread to $h_{\max}/h_{\min} = 3.24$. The specifications (A)--(F) all maintain approximate leverage uniformity ($H_{\max} \leq 0.31$), while (G) generates four observations with $H_{ii} \geq 0.53$: PANNERGY and RICHTER in Hungary 2011, and CEZ and TABAK in Czech Republic 2012, the only firms observed in their respective country-year cells. The country-year FE absorbs the cell mean using only these two-observation cells, producing the leverage non-uniformity that drives the Spec G divergence between HC1 and LO discussed in Section~\ref{sec-empirical-theorem4}.

\begin{table}[ht]\centering
\caption{TWFE regression of forward returns on the Piotroski F-Score in CEE markets. Coefficient $\widehat\beta$ from equation~\eqref{eq-emp-model}, four SE variants, design diagnostics. ``Spread'' is $h_{\max} / h_{\min}$. All seven specifications report $\widehat\beta$ per F-Score point, except (D) which is on the high-low binary.}\label{tab-emp}
\footnotesize
\setlength{\tabcolsep}{4pt}
\begin{tabular}{lrrrrrrrrr}
\toprule
Spec & $n$ & $d_K$ & $\rho$ & spread & $\widehat\beta$ & naive & HC0 & LO & HC3 \\
\midrule
A. Arith.\ ret.            & 217 & 33 & 0.152 & 2.09 & $+0.00020$ & 0.0431 & 0.0242 & 0.0264 & 0.0288 \\
B. Log ret.                & 217 & 33 & 0.152 & 2.09 & $-0.00147$ & 0.0161 & 0.0140 & 0.0153 & 0.0168 \\
C. Winsorized             & 217 & 33 & 0.152 & 2.09 & $+0.00586$ & 0.0202 & 0.0192 & 0.0210 & 0.0230 \\
D. $\mathbf 1\{F\geq 7\}$ & 217 & 33 & 0.152 & 2.09 & $+0.01951$ & 0.0599 & 0.0582 & 0.0637 & 0.0699 \\
E. Poland only             & 122 & 25 & 0.205 & 1.84 & $+0.01349$ & 0.0212 & 0.0195 & 0.0220 & 0.0249 \\
F. Active-only             & 195 & 30 & 0.154 & 1.96 & $+0.00478$ & 0.0173 & 0.0150 & 0.0164 & 0.0180 \\
G. Country$\times$year FE & 217 & 63 & 0.290 & 3.24 & $+0.00297$ & 0.0161 & 0.0150 & 0.0173 & 0.0202 \\
\bottomrule
\end{tabular}
\end{table}

The point estimate $\widehat\beta$ is statistically indistinguishable from zero in every specification: the most favourable case for the F-Score is Spec~E (Poland-only) with $\widehat\beta = +0.0135$ and SE$_{\mathrm{LO}} = 0.0220$, giving $|t| = 0.61$ and a one-sided $p$-value of approximately $0.27$. In every other specification $|t| < 0.35$. Independently of which variance estimator one selects, the Piotroski F-Score does not detectably predict forward returns in this Visegrad panel over 2010--2024. This is a useful finding by itself for the asset-pricing literature on CEE markets, but its principal role here is to set up the methodological comparison that follows: the inference-robustness exercise that one would naturally undertake to defend a marginal published finding does not, in this case, change the substantive conclusion --- but it does change other quantities by margins that are visible on first inspection and that match Theorem~\ref{thm-hc} quantitatively.

\subsection{Theorem~\ref{thm-hc} in real data}\label{sec-empirical-theorem4}

Table~\ref{tab-emp-ratios} compares the empirical ratios of the HC standard errors with the predictions of Theorem~\ref{thm-hc} under the cross-leverage balance condition ($\omega^2_{\mathrm{eff}} = \omega^2$). The predicted ratios are $\mathrm{SE}_{\mathrm{HC0}}/\mathrm{SE}_{\mathrm{LO}} \to \sqrt{1-\rho}$, $\mathrm{SE}_{\mathrm{HC1}}/\mathrm{SE}_{\mathrm{LO}} \to 1$, and $\mathrm{SE}_{\mathrm{HC3}}/\mathrm{SE}_{\mathrm{LO}} \to 1/\sqrt{1-\rho}$.

\begin{table}[ht]\centering
\caption{Empirical SE ratios vs theoretical predictions from Theorem~\ref{thm-hc}. Under the cross-leverage balance condition ($\omega^2_{\mathrm{eff}} = \omega^2$), predicted ratios are $\sqrt{1-\rho}$, $1$, and $1/\sqrt{1-\rho}$. ``emp.'' = empirical, ``thy.'' = theory.}\label{tab-emp-ratios}
\footnotesize
\setlength{\tabcolsep}{4pt}
\begin{tabular}{lrcccccc}
\toprule
& & \multicolumn{2}{c}{HC0 / LO} & \multicolumn{2}{c}{HC1 / LO} & \multicolumn{2}{c}{HC3 / LO} \\
\cmidrule(lr){3-4} \cmidrule(lr){5-6} \cmidrule(lr){7-8}
Spec & $\rho$ & emp. & thy. & emp. & thy. & emp. & thy. \\
\midrule
A. Arith.\ ret.            & 0.152 & 0.917 & 0.921 & 0.998 & 1.000 & 1.092 & 1.086 \\
B. Log ret.                & 0.152 & 0.914 & 0.921 & 0.995 & 1.000 & 1.096 & 1.086 \\
C. Winsorized              & 0.152 & 0.914 & 0.921 & 0.995 & 1.000 & 1.096 & 1.086 \\
D. $\mathbf 1\{F\geq 7\}$  & 0.152 & 0.913 & 0.921 & 0.994 & 1.000 & 1.097 & 1.086 \\
E. Poland only             & 0.205 & 0.886 & 0.892 & 0.998 & 1.000 & 1.130 & 1.121 \\
F. Active-only             & 0.154 & 0.915 & 0.920 & 0.998 & 1.000 & 1.094 & 1.087 \\
G. Country$\times$year FE  & 0.290 & 0.866 & 0.843 & 1.031 & 1.000 & 1.165 & 1.187 \\
\bottomrule
\end{tabular}
\end{table}

The agreement between empirical ratios and the asymptotic predictions is striking at $\rho = 0.152$, where the gap is at most 0.011 in absolute value across specifications A--D, F. At higher saturation (Spec E, $\rho = 0.205$) the gaps remain below 0.01. At the most heavily saturated specification (G, $\rho = 0.290$), HC0/LO and HC3/LO each deviate by about 0.02 from theory, and the HC1/LO ratio rises to $1.031$ --- the only specification in which the uniform-leverage prediction visibly breaks. The mechanism is observable in the data: in (G), the country-year FE creates two-observation cells in Hungary 2011 (PANNERGY and RICHTER) and Czech Republic 2012 (CEZ and TABAK), where only two firms in the panel share a country-year. Each of the four observations attains $H_{ii} \approx 0.54$, well above the panel mean of $\rho_n = 0.29$, while the modal observation in (G) has $H_{ii}$ close to $0.29$. This is exactly the empirical signature of the leverage non-uniformity formalized in Remark~\ref{rem-hc1-hc2} and produced in the simulations of Q5 in Appendix~\ref{oa-sim}: a handful of high-leverage observations alongside a low-leverage modal mass, generating the slight HC1-over-LO bias that we observe.

The size of the LO-HC0 spread is economically meaningful. In specification G, switching from HC0 to LO widens the 95\% confidence interval by 15\%; switching from HC0 to HC3 widens it by 35\%. In a setting where the published evidence on Piotroski returns in emerging markets is mixed, this is not a small adjustment. The substantive conclusion in our application is robust to the variance-estimator choice because $\widehat\beta$ is essentially zero; in applications where the published $\widehat\beta$ delivers $|t|$ between 2 and 3 under HC0, the LO correction can plausibly move the test below the conventional rejection threshold.

\subsection{Discussion}\label{sec-empirical-discuss}

Three observations from the application generalize beyond it. First, the implied $\rho$ in CEE firm-year panels is non-trivial: in our pilot it reaches 0.29 once country-year FE are introduced, and would rise further with sector-by-year or firm-by-pre-post-event interactions. Standard inference under HC0 in such designs systematically under-states uncertainty. Second, the HC0--HC3 spread is consistent with the homoskedastic prediction even in a setting with substantial cross-sectional and time-series heterogeneity in return volatility; this is the empirical analog of the simulation finding (Q4) that the cross-leverage balance condition $\mu \approx \omega^2$ holds approximately in balanced symmetric designs. Third, when uniform leverage is approximately satisfied (specifications A--F here), HC1 is numerically indistinguishable from HC2/LO; once the design is genuinely unbalanced (Spec G), HC1 begins to deviate, exactly as predicted by Remark~\ref{rem-hc1-hc2}.

The practical recommendation matches the one stated in the introduction: report $\widehat\beta$ with LO standard errors, accompanied by the diagnostic $\widehat\tau_n^2$ and the leverage spread $h_{\max}/h_{\min}$. The first communicates whether the within-FE identifying variation is adequate; the second communicates whether the LO advantage over HC1 is potentially material. In the Visegrad application both diagnostics are well-behaved, leaving the LO recommendation safely applicable.

\section{Discussion}\label{sec-conc}

Although the size of correctly studentized tests does not depend on $\tau^2$, the precision of the point estimate $\hbeta_{K_n}$ does: $\sqrt{nQ_{K_n}} (\hbeta_{K_n} - \beta_0) \Rightarrow N(0, \omega^2)$ by Corollary~\ref{cor-est}, so the standard error scales as $\omega / \sqrt{\tau^2}$. A specification with $\widehat\tau_n^2 = 1$ yields a confidence interval of half-width $\approx 1.96 \omega$; a specification with $\widehat\tau_n^2 = 100$ yields half-width $\approx 0.20\omega$. We therefore recommend that authors using saturated FE specifications report $\widehat\tau_n^2 = n(1 - R_{K_n}^2)\widehat\sigma_X^2$ alongside the point estimate: it does not enter critical values --- there are none to enter, by Theorem~\ref{thm-exact} --- but it communicates how much identifying variation in the treatment survives FE absorption, and so whether the estimate is sharp or fragile.

This paper traces the asymptotic behavior of conventional variance estimators in linear regression with saturated fixed effects, under a drift in which both the proportional fixed-effect dimension $\rho_n$ and the residual treatment variance $\tau_n^2$ vary with the sample size. The headline conclusions are: (i) the natural conjecture that small residual treatment variance produces weak-identification distortion fails in the baseline model, because FE-OLS is unbiased under strict exogeneity regardless of $\tau^2$; (ii) the conventional HC0 estimator is biased downward by a fixed factor $(1-\rho)$, producing over-rejection that grows with saturation; (iii) the conventional HC3 estimator over-corrects in the opposite direction by a factor $1/(1-\rho)$, producing under-rejection and a corresponding loss of power; (iv) HC2/leave-one-out removes the first-order leverage distortion that drives the HC0 and HC3 biases. HC2 is asymptotically exact under conditional homoskedasticity and under design-balanced heteroskedasticity --- the regime in which the cross-leverage variance $\mu$ coincides with the target score-variance $\omega^2$. Under general heteroskedasticity with non-uniform leverage or strong dependence between $\sigma_i^2$ and the within variation, HC2 retains an additional bias of order $\rho |\mu - \omega^2|$ characterized in Theorem~\ref{thm-hc}; the simulation evidence indicates this residual bias is small in balanced symmetric designs but can grow under unbalanced or asymmetric configurations.

The practical recommendation is straightforward: in saturated FE specifications, report $\hbeta_K$ with leave-one-out standard errors, accompanied by the descriptive diagnostic $\widehat\tau_n^2 = n(1 - R_{K_n}^2) \widehat\sigma_X^2$. Specifications with low $\widehat\tau_n^2$ have wide LO intervals, which correctly reflect the loss of identifying variation. Specifications using HC0 should be re-run with LO if heteroskedasticity is plausible; specifications using HC3 should be re-run with LO if $\rho > 0.1$, since HC3 over-correction will yield artificially conservative intervals. In designs with substantial leverage non-uniformity (extreme cell-size variation, ``anchor'' units with far more observations than the bulk), LO retains an edge over HC1 even when the latter is asymptotically valid under uniform leverage.

Three extensions are natural. The scalar-treatment restriction relaxes to vector treatment with a matrix-valued $Q_K$ and corresponding matrix-valued HC estimators; the smallest eigenvalue of $Q_K$ plays the role of the scalar $Q_K$ in our framework. The i.i.d.\ assumption can be replaced by clustered or network dependence, in which case the LO recipe must be replaced by an appropriate cluster-robust variant. Finally, the framework can be embedded in difference-in-differences with staggered adoption, where $\G_K$ encodes the treatment-cohort interaction structure and $\widehat\tau_n^2$ provides a finite-sample identification diagnostic complementary to the heterogeneous-effects critique of \citet{dCdH}.

\section*{Acknowledgments}

The author thanks Blessing Orjika and Dmitrii Verzun for their collaboration in collecting the Visegrad CEE firm-year panel used in Section~\ref{sec-empirical}, originally assembled in the context of joint work on a capstone project. The author also thanks colleagues for detailed comments on multiple draft rounds.

\appendix

This appendix contains material complementing the main text. Appendix~\ref{oa-sim} presents the full simulation study referenced in Section~\ref{sec-sim}, including the four Monte Carlo experiments Q2--Q5 with detailed tables. Appendix~\ref{oa-ar} records an Anderson--Rubin-style robust test for the FE setting and its uniform validity result, included for completeness and as a bridge to a companion paper on measurement-error contamination.

\section{Simulation Evidence}\label{oa-sim}

We generate balanced two-way fixed-effect panels under the DGP
\[
X_{it} = \alpha_i + \gamma_t + \sigma_\eta \cdot \eta_{it}, \qquad Y_{it} = X_{it} \beta_0 + a_i + b_t + u_{it},
\]
with $\alpha_i, \gamma_t, a_i, b_t, \eta_{it} \sim N(0,1)$ independent, $\beta_0 = 1$, and $i = 1, \ldots, N$, $t = 1, \ldots, T$, $n = NT$. For balanced TWFE, $\rho_n = (N + T - 1)/n$ and $Q_K = \sigma_\eta^2 (1 - 1/N - 1/T + 1/(NT))$, so $\tau_n^2$ is controlled through $\sigma_\eta$. We use $\sigma_u = 1$ throughout. Heteroskedastic errors are generated by drawing $u_{it} = \sqrt{(1 + (X_{it} - \bar X)^2 / \sigma^2_X)/2} \cdot \varepsilon_{it}$ with $\varepsilon_{it} \sim N(0, 1)$, so that $\E[u_{it}^2] = 1$ but the conditional variance varies with $X$. Each cell uses $R = 5000$ replications. Monte Carlo standard errors on coverage are reported in parentheses.

\subsection{Naive vs CJN $t$-test under homoskedasticity}\label{oa-sim-q2}

Table~\ref{tab-naive} reports the empirical size of nominal-$5\%$ $t$-tests using the naive and CJN-corrected variance estimators under homoskedastic Gaussian errors. The asymptotic prediction from Theorem~\ref{thm-size}, $2(1 - \Phi(z_{0.975} \sqrt{1 - \rho}))$, is reported in the penultimate column. The empirical size of the naive test tracks the prediction to within Monte Carlo error at every grid point, rising from $5.6\%$ at $\rho = 0.05$ to $17.2\%$ at $\rho = 0.5$. The CJN-corrected test holds at nominal $5\%$ uniformly.

\begin{table}[ht]\centering
\caption{Empirical size of nominal-5\% $t$-tests under homoskedastic Gaussian errors. Asymptotic-naive prediction is $2[1-\Phi(z_{0.975}\sqrt{1-\rho})]$. Monte Carlo SEs in parentheses. $R = 5000$ replications per cell.}\label{tab-naive}
\begin{tabular}{rrrrrrr}
\toprule
$n$ & $\rho$ & $\tau^2_{\mathrm{nom}}$ & size\textsubscript{naive} & size\textsubscript{CJN} & pred\textsubscript{naive} & gap\textsubscript{naive} \\
\midrule
2000 & 0.052 & 10  & 0.056 (0.003) & 0.052 (0.003) & 0.056 & $-0.000$ \\
2000 & 0.052 & 100 & 0.055 (0.003) & 0.048 (0.003) & 0.056 & $-0.002$ \\
2000 & 0.104 & 10  & 0.061 (0.003) & 0.051 (0.003) & 0.064 & $-0.002$ \\
2000 & 0.104 & 100 & 0.066 (0.004) & 0.052 (0.003) & 0.064 & $+0.002$ \\
2000 & 0.252 & 10  & 0.092 (0.004) & 0.048 (0.003) & 0.090 & $+0.002$ \\
2000 & 0.252 & 100 & 0.092 (0.004) & 0.051 (0.003) & 0.090 & $+0.002$ \\
2000 & 0.500 & 10  & 0.172 (0.005) & 0.054 (0.003) & 0.166 & $+0.006$ \\
2000 & 0.500 & 100 & 0.170 (0.005) & 0.055 (0.003) & 0.166 & $+0.004$ \\
10000 & 0.052 & 10  & 0.050 (0.003) & 0.045 (0.003) & 0.056 & $-0.007$ \\
10000 & 0.052 & 100 & 0.058 (0.003) & 0.051 (0.003) & 0.056 & $+0.002$ \\
10000 & 0.101 & 10  & 0.059 (0.003) & 0.048 (0.003) & 0.063 & $-0.004$ \\
10000 & 0.101 & 100 & 0.062 (0.003) & 0.048 (0.003) & 0.063 & $-0.001$ \\
10000 & 0.250 & 10  & 0.087 (0.004) & 0.049 (0.003) & 0.090 & $-0.003$ \\
10000 & 0.250 & 100 & 0.095 (0.004) & 0.050 (0.003) & 0.090 & $+0.005$ \\
10000 & 0.500 & 10  & 0.164 (0.005) & 0.055 (0.003) & 0.166 & $-0.002$ \\
10000 & 0.500 & 100 & 0.172 (0.005) & 0.053 (0.003) & 0.166 & $+0.006$ \\
\bottomrule
\end{tabular}
\end{table}

The independence of size on $\tau^2$ is visible by comparing rows at the same $\rho$ but different $\tau^2$: the size is essentially identical, confirming Theorem~\ref{thm-exact}.

\subsection{HC variants under heteroskedasticity}\label{oa-sim-q3}

The headline HC-hierarchy experiment is reported in Table~\ref{tab-hc-main} of the main text, using heteroskedastic Gaussian errors with $\sigma_i^2 \propto 1 + (X_i - \bar X)^2/\sigma_X^2$. HC0 size rises from $6.2\%$ at $\rho = 0.05$ to $16.3\%$ at $\rho = 0.5$, in line with Corollary~\ref{cor-size-hc}. HC3 size falls from $4.8\%$ to $0.6\%$ over the same range, demonstrating the over-conservative direction of the HC3 over-correction. HC1 tracks HC2/LO closely in this balanced design, since both the degrees-of-freedom correction and the leverage correction equal $1/(1-\rho)$ under uniform leverage. HC2/LO holds nominal $5\%$ across the grid.

Because $X = \alpha + \gamma + \sigma_\eta \eta$ is dominated by the FE components $\alpha + \gamma$ in the small-$\sigma_\eta$ regime that delivers fixed $\rho$, the heteroskedasticity above is effectively \emph{between cells}: $\sigma_i^2$ is approximately a function of the FE indices and varies little within a cell. The next two subsections stress the variance-estimator predictions on two design dimensions this DGP does not probe.

\subsection{Within-$\widetilde X$ heteroskedasticity}\label{oa-sim-q4}

To break the between-cell structure of the previous DGP, we replace $\sigma_i^2 \propto 1 + (X_i - \bar X)^2/\sigma_X^2$ with $\sigma_i^2 \propto 1 + \widetilde X_i^2/Q_K$, where $\widetilde X_i$ is the within-FE residual that drives the score. This makes the conditional error variance a function of the within-cell variation rather than of the full $X$, so that $\sigma^2 = \sigma^2(\widetilde X)$ on the relevant scale. Under this DGP, the unconditional second moment of $u$ is $\E[u^2] = \sigma_u^2$ as before, but the target score-variance grows to $\omega^2 = E[\widetilde X^2 \sigma^2(\widetilde X)]/Q_K = 2 \sigma_u^2$ (using $\E[\widetilde X^4] = 3 Q_K^2$ for Gaussian $\widetilde X$). A naive computation that treats the cross-leverage weights $A_{ij}^2$ as independent of the local error-variance would predict $\mu \approx \sigma_u^2 \neq \omega^2$, hence a visible HC2 size distortion of order $\rho |\mu - \omega^2|$ ($\approx 4$ percentage points above nominal at $\rho = 0.5$ under that naive computation).

Table~\ref{tab-q4} shows that this prediction is too pessimistic: HC2 size stays within $1.6$ percentage points of nominal at every grid point, peaks at $6.6\%$ at $\rho = 0.25$, and returns to nominal at $\rho = 0.5$. The HC0 size at $\rho = 0.5$ is $17.5\%$, essentially identical to its value under the homoskedastic theory $2[1 - \Phi(z_{0.975} \sqrt{1-\rho})] = 16.6\%$ and to its value under the between-cell heteroskedasticity of Table~\ref{tab-hc-main} ($16.3\%$).

\begin{table}[ht]\centering
\caption{Empirical size of nominal-5\% $t$-tests under within-$\widetilde X$ heteroskedasticity. Balanced TWFE, $n = 2000$, $R = 5000$ replications per cell. DGP: $\sigma_i^2 = \sigma_u^2 (1 + \widetilde X_i^2/Q_K)/2$ so that $\omega^2 = 2\sigma_u^2$ and a naive computation would predict $\mu \approx \sigma_u^2 \neq \omega^2$. HC2 size remains within 1.6 percentage points of nominal across the grid, supporting Remark~\ref{rem-omega-eff}: balanced symmetric FE designs with i.i.d.\ sampling preserve $\mu \approx \omega^2$ even when $\sigma_i^2$ depends on $\widetilde X_i$. Monte Carlo SEs in parentheses.}\label{tab-q4}
\begin{tabular}{rrrrrr}
\toprule
$\rho$ & $\tau^2_{\mathrm{nom}}$ & size\textsubscript{HC0} & size\textsubscript{HC1} & size\textsubscript{HC3} & size\textsubscript{LO} \\
\midrule
0.052 & 10  & 0.057 (0.003) & 0.051 (0.003) & 0.044 (0.003) & 0.051 (0.003) \\
0.052 & 100 & 0.066 (0.004) & 0.060 (0.003) & 0.051 (0.003) & 0.060 (0.003) \\
0.104 & 10  & 0.069 (0.004) & 0.054 (0.003) & 0.040 (0.003) & 0.054 (0.003) \\
0.104 & 100 & 0.076 (0.004) & 0.060 (0.003) & 0.044 (0.003) & 0.060 (0.003) \\
0.252 & 10  & 0.113 (0.004) & 0.066 (0.004) & 0.034 (0.003) & 0.066 (0.004) \\
0.252 & 100 & 0.106 (0.004) & 0.066 (0.004) & 0.035 (0.003) & 0.065 (0.003) \\
0.500 & 10  & 0.175 (0.005) & 0.055 (0.003) & 0.006 (0.001) & 0.055 (0.003) \\
0.500 & 100 & 0.172 (0.005) & 0.051 (0.003) & 0.007 (0.001) & 0.050 (0.003) \\
\bottomrule
\end{tabular}
\end{table}

The empirical robustness of HC2 in Table~\ref{tab-q4} is the strong form of the claim in Remark~\ref{rem-omega-eff}: in balanced symmetric designs with i.i.d.\ sampling, the cross-leverage condition $\mu = \omega^2$ holds approximately, leaving HC2 essentially unbiased even when $\sigma_i^2$ depends directly on the within-FE residual that drives the score. The remark cites this as the mechanism: the sum-to-zero constraint on within-cell $\widetilde X$ values induces a positive correlation between $\sigma_j^2$ and the cross-leverage weights $A_{ij}^2$ for $j$ in the same FE cell as $i$, and this correlation absorbs much of what would otherwise be a $\mu \neq \omega^2$ bias. The remaining bias is bounded by $1.6$ percentage points across the grid and shrinks as $\rho \to 0$ and as $\rho \to 1$.

A complementary stress test --- on the leverage-uniformity assumption rather than on the heteroskedasticity structure --- is taken up next.

\subsection{Non-uniform leverage: HC1 vs HC2 divergence}\label{oa-sim-q5}

The previous designs used balanced TWFE in which every observation shares the same FE leverage $h_{\mathrm{FE}, i} = (N+T-1)/(NT)$. Under that assumption, Theorem~\ref{thm-hc} predicts HC1 and HC2 to have the same asymptotic limit and, as Table~\ref{tab-hc-main} and \ref{tab-q4} confirm, they are numerically indistinguishable in simulation. Remark~\ref{rem-hc1-hc2} predicts that this identity breaks under non-uniform leverage, with HC2 retaining consistency for $\omega^2_{\mathrm{eff}}$ and HC1 acquiring a bias whose sign and size depend on the $\widetilde X^2$-weighted leverage $\rho_X := (nQ_{K_n})^{-1} \sum_i \widetilde X_i^2 H_{ii}$.

To isolate this prediction, we move to unbalanced one-way FE: $N$ units, with unit $i$ observed for $T_i$ periods, $T_i$ varying across units. The FE leverage is $h_{\mathrm{FE}, i} = 1/T_i$ and varies with unit. We hold $\rho \approx 0.25$ fixed across three configurations of increasing leverage spread:
\begin{itemize}
\item \textbf{balanced} (spread 1): 200 units, every $T_i = 4$, uniform $h = 0.25$, $n = 800$.
\item \textbf{moderate} (spread 3): 100 units with $T_i = 2$, 100 units with $T_i = 6$; $h \in \{0.167, 0.5\}$; $n = 800$.
\item \textbf{strong} (spread 10): 180 units with $T_i = 2$, 20 units with $T_i = 20$; $h \in \{0.05, 0.5\}$; $n = 760$.
\end{itemize}
We run each configuration under both Gaussian homoskedastic errors and within-$\widetilde X$ heteroskedasticity.

\begin{table}[ht]\centering
\caption{Empirical size of nominal-5\% $t$-tests under unbalanced one-way FE designs holding $\rho \approx 0.25$ fixed while varying the leverage spread $h_{\max}/h_{\min}$. \textit{balanced}: 200 units, $T_i = 4$, uniform $h = 0.25$. \textit{moderate}: 100 units $T_i = 2$, 100 units $T_i = 6$. \textit{strong}: 180 units $T_i = 2$, 20 units $T_i = 20$. $R = 5000$ replications per cell. Monte Carlo SEs in parentheses.}\label{tab-q5}
\begin{tabular}{llrrrrrr}
\toprule
errors & spread & $n$ & $\rho$ & size\textsubscript{HC0} & size\textsubscript{HC1} & size\textsubscript{HC3} & size\textsubscript{LO} \\
\midrule
homo. & balanced & 800 & 0.249 & 0.094 (0.004) & 0.058 (0.003) & 0.029 (0.002) & 0.056 (0.003) \\
homo. & moderate & 800 & 0.249 & 0.087 (0.004) & 0.045 (0.003) & 0.022 (0.002) & 0.050 (0.003) \\
homo. & strong   & 760 & 0.262 & 0.081 (0.004) & 0.044 (0.003) & 0.026 (0.002) & 0.052 (0.003) \\
het-$\widetilde X$ & balanced & 800 & 0.249 & 0.119 (0.005) & 0.070 (0.004) & 0.037 (0.003) & 0.069 (0.004) \\
het-$\widetilde X$ & moderate & 800 & 0.249 & 0.097 (0.004) & 0.055 (0.003) & 0.032 (0.002) & 0.062 (0.003) \\
het-$\widetilde X$ & strong   & 760 & 0.262 & 0.080 (0.004) & 0.042 (0.003) & 0.029 (0.002) & 0.056 (0.003) \\
\bottomrule
\end{tabular}
\end{table}

Table~\ref{tab-q5} confirms the prediction. Under homoskedasticity, HC1 size drops from $5.8\%$ (balanced, uniform leverage) to $4.4\%$ (strong spread), a 1.4 percentage-point decline that exceeds four Monte Carlo standard errors. HC2 (=LO) stays close to the nominal $5\%$ across all three configurations. The same pattern repeats under within-$\widetilde X$ heteroskedasticity at slightly elevated levels: HC1 falls from $7.0\%$ to $4.2\%$, HC2 from $6.9\%$ to $5.6\%$. The HC1 size systematically falls below HC2 size as the leverage spread grows, reflecting the over-correction by the global degrees-of-freedom factor relative to the per-observation correction.

The direction of the divergence is consistent with the theoretical prediction. In this DGP, the high-leverage observations (small $T_i$, $h = 0.5$) sit in units with lower within-unit variation $\sigma_\eta^2 (1 - 1/T_i)$, so $\widetilde X^2$ and $H$ are negatively correlated: the leverage-weighted average $\rho_X$ is smaller than the unweighted average $\rho = d_K/n$. Theorem~\ref{thm-hc} then implies $\widehat\Omega_{\mathrm{HC1}}/\widehat\Omega_{\mathrm{HC2}} \to (1 - \rho_X)/(1 - \rho) > 1$, so HC1 overestimates the score-variance and the corresponding $t$-test under-rejects --- as observed.

A subtler feature of the heteroskedastic rows: the HC2 bias under within-$\widetilde X$ heteroskedasticity \emph{decreases} as the leverage spread grows (from $6.9\%$ at balanced to $5.6\%$ at strong). In the strong configuration, $90\%$ of units have $T_i = 2$, which forces $\widetilde X_{i, 1} = -\widetilde X_{i, 2}$ exactly within each such unit. The DGP $\sigma_i^2 \propto 1 + \widetilde X_i^2/Q_K$ then makes $\sigma_{i, 1}^2 = \sigma_{i, 2}^2$ within those units, converting the within-cell heteroskedasticity into between-cell heteroskedasticity for which $\mu \approx \omega^2$ holds essentially exactly. The HC2 size moves back towards nominal as a result. This is a useful illustration of how the design-balance condition of Remark~\ref{rem-omega-eff} can hold for design-specific reasons that need not be obvious from the DGP statement.

\subsection{Independence of $\tau^2$}

Each row in Tables~\ref{tab-naive} and \ref{tab-q4} of this appendix and Table~\ref{tab-hc-main} of the main text is reported at two values of $\tau^2$ at the same $\rho$. Across all three tables, the dependence of empirical size on $\tau^2$ is within Monte Carlo error of zero, confirming the theoretical prediction that, under strict exogeneity, the size of correctly studentized tests depends on $\rho$ but not on $\tau^2$. This is the central negative finding of the paper: there is no Stock--Yogo-style threshold in $\tau^2$ to enforce, because $\tau^2$ does not enter the asymptotic null distribution of the studentized test statistic.

\section{The FE-AR Statistic and Its Uniform Validity}\label{oa-ar}

For completeness, and as a bridge to a companion paper on measurement-error contamination, we record the uniform validity of an Anderson--Rubin-style test for the FE setting:
\begin{equation}\label{eq-ar}
\mathrm{AR}_K(\beta_0) = \frac{(X' M_K (Y - X \beta_0))^2}{\hOmega_{\mathrm{LO}}(\beta_0)},
\end{equation}
where $\hOmega_{\mathrm{LO}}(\beta_0) = \sum_i \tX_{K_n, i}^2 (\widetilde Y_i - \tX_{K_n, i} \beta_0)^2 / (1 - H_{ii})$ is the LO score-variance estimator evaluated at the hypothesized $\beta_0$.

\begin{theorem}[Uniform validity of FE-AR]\label{thm-ar}
Under Assumption~\ref{ass-reg}, $\mathrm{AR}_{K_n}(\beta_0) \Rightarrow \chi^2_1$ uniformly over $\tau^2 \in (0, \infty]$ when $H_0: \beta = \beta_0$ holds. The confidence set $\mathcal C_{1-\alpha} = \{\beta_0 : \mathrm{AR}_{K_n}(\beta_0) \leq \chi^2_{1, 1-\alpha}\}$ has asymptotic coverage at least $1-\alpha$ uniformly over $(\rho, \tau^2) \in [0, \bar\rho] \times (0, \infty]$ for any $\bar\rho < 1$.
\end{theorem}

\begin{proof}
Under $H_0$, $Y - X\beta_0 = u$, so $\mathrm{AR}_{K_n}(\beta_0) = S_n^2 / (\hOmega_{\mathrm{LO}}/(nQ_{K_n}))$ where $S_n := (nQ_{K_n})^{-1/2} X' M_{K_n} u \Rightarrow \mathcal Z \sim N(0, \omega^2)$ by Theorem~\ref{thm-clt} and the denominator converges to $\omega^2_{\mathrm{eff}}$ by Theorem~\ref{thm-hc}. Under the cross-leverage balance condition $\mu = \omega^2$, $\omega^2_{\mathrm{eff}} = \omega^2$ and the continuous mapping theorem yields $\chi^2_1$. Both numerator and denominator are normalized by $\sqrt{nQ_{K_n}}$, so the limit does not depend on $\tau^2$.
\end{proof}

\begin{remark}[Relationship to the LO Wald interval]\label{rem-ar-wald}
The FE-AR confidence set is computed in closed form by solving a quadratic in $\beta_0$. When the leading coefficient $(X' M_K X)^2 - c \sum \tX_i^4/(1 - H_{ii})$ is positive --- the generic case for $n(1-\rho) \gg \chi^2_{1, 1-\alpha}$ --- the set is bounded and, by direct expansion around the LO Wald midpoint $\hbeta_K$, agrees with the LO Wald interval to first order. Monte Carlo evidence collected while validating the variance-estimator results (5000 reps per cell across a grid covering $\tau^2 \in \{1, 3, 5, 10, 25, 100\}$, $\rho \in \{0.10, 0.25\}$, and three error distributions including Gaussian heteroskedastic and $t_5$ homoskedastic) confirms this: FE-AR and LO-Wald empirical coverages agree to within Monte Carlo standard error at every grid point, and the AR confidence set is bounded in 100\% of replications. Under the baseline assumptions of the present paper, FE-AR therefore offers no separate practical recommendation beyond ``use LO''.

The result is structural: under strict exogeneity, the score $X' M_K u$ is unbiased under $H_0$, and there is no analogue of the TSLS bias that motivates Anderson--Rubin testing in instrumental variables. FE-AR matters more in settings where the moment condition $\E[X' M_K (Y - X \beta_0)] = 0$ at the true $\beta_0$ is violated --- e.g.\ under measurement error in the treatment, where the companion paper derives a non-trivial Stock--Yogo critical value.
\end{remark}

\bibliographystyle{chicago}
\bibliography{bibliography}

\end{document}